\documentclass[reqno,11pt]{amsart}
\usepackage{amsmath,amssymb}
\usepackage{color}

\newtheorem{theorem}{Theorem}[section]

\newtheorem{lemma}[theorem]{Lemma}
\newtheorem{proposition}[theorem]{Proposition}
\theoremstyle{definition}
\newtheorem{definition}[theorem]{Definition}
\theoremstyle{remark}
\theoremstyle{remark}
\newtheorem*{assumption}{Assumption}
\newtheorem{remark}[theorem]{Remark}
\numberwithin{equation}{section}

\allowdisplaybreaks[3]
\email{alex.daletskii@york.ac.uk}
\email{kondrat@math.uni-bielefeld.de}
\email{jkozi@hektor.umcs.lublin.pl}

\begin{document}
\title[Phase Transitions in Continuum Ferromagnets]{Phase Transitions in
Continuum Ferromagnets with Unbounded Spins}
\author{ Alexei Daletskii}
\address{Department of Mathematics, University of York, York YO10 DD, UK}
\author{Yuri Kondratiev}
\address{Fakut\"at f\"ur Mathematik, Universit\"at Bielefeld, Bielefeld
D-33615, Germany}
\author{ Yuri Kozitsky}
\address{Instytut Matematyki, Uniwersytet Marii Curie-Sk{\l }odowskiej,
20-031 Lublin, Poland}

\begin{abstract}
States of thermal equilibrium of an infinite system of interacting particles
in $\mathbb{R}^{d}$ are studied. The particles bear `unbounded' spins with a
given symmetric a priori distribution. The interaction between the particles
is pairwise and splits into position-position and spin-spin parts. The
position-position part is described by a superstable potential, and the
spin-spin part is attractive and of finite range. Thermodynamic states of
the system are defined as tempered Gibbs measures on the space of marked
configurations. It is proved that the set of such measures contains at least
two elements if the activity is big enough.
\end{abstract}

\maketitle

%\tableofcontents

\section{Introduction}

%\label{S1}

\subsection{Posing the problem}

%\label{SS1.1}
The mathematical theory of thermal equilibrium of infinite particle systems
relies on the use of conditional probabilities, see \cite{Geor,Pre,Simon},
by means of which one defines the set of Gibbs measures that exist at given
values of the model parameters. The multiplicity of such measures is then
interpreted as the possibility for the system to undergo a phase transition
and is one of the most fundamental aspects of the theory. Historically, the
Gibbsian formalism was first developed for the Ising spin model, where each
`particle' was associated with a point $x\in \mathbb{Z}^{d}$ and can be in
one of two states, cf. \cite{Do68,Do70}. This is the simplest model of a
crystalline magnet. It took, however, eight years (since the publication of
first Dobrushin's papers) until the Gibbs states of lattice models with
`unbounded' spins were constructed in \cite{LP} by means of new tools
developed during that time. In noncrystalline magnets, the particles are
distributed over a continuous medium (e.g., $\mathbb{R}^{d}$), and their
positions may not be fixed. The corresponding physical substances are e.g.
magnetic gases, ferrofluids, amorphous magnets, etc., see \cite{Gruber} for
further information on this issue. For a ferrofluid with hard core repulsion
and Ising spins, the existence of spontaneous magnetization was proved in
\cite{Gruber}, which later on was extended in \cite{RZ} to similar models
with continuous bounded spins. The results of both these works can be
interpreted as the proof of the multiplicity of the corresponding Gibbs
measures provided their existence is established. In \cite{GeHa}, the
existence and multiplicity of Gibbs measures were proved for the model in
which each particle can be in one of $q$ states -- continuum Potts model.
Our aim is to elaborate the theory of phase transitions in systems of
particles in continuum carrying `unbounded' spins. To this end we employ the
latest developments in the theory of Gibbs measures with irregular
underlying sets \cite{DKKP1,KKP1,KKP2,KPR,KP} combined with contemporary
methods of the analysis on configuration spaces \cite{AKLU,KKun,KKut,KunaPhD}%
.

There are two different approaches to studying continuum systems of
particles with spins: (a) the positions of the particles are taken at random
from an ensemble characterized by a given probability law, and the spins are
distributed according to a random `spin-only' Gibbs measure; (b) the
interaction between the particles contains spin-spin and position-position
parts and the joint probability distribution is given by a `position and
spin' Gibbs measure. Phase transitions in the systems of the first type
(quenched magnets with Poisson-distributed positions) have been considered
in \cite{DKKP}. In the present paper, we study a system of the second type,
with the position-position interactions satisfying the strong superstability
condition, cf. \cite{Ru70,RT,RT1,KPR}. Our main technical tool is the finite
volume reduction to a quenched system and the use of the percolation theory,
in the spirit of \cite{GeHa} and \cite{H}.

\subsection{The paper overview}

%\label{SS1.2}
We consider the following infinite-particle model. Each particle is
characterized by position $x\in X=\mathbb{R}^{d}$, $d\geq 1$, and spin $%
\sigma \in S=\mathbb{R}$. The particles interact via a pair interaction
potential of the form
\begin{equation}
\Psi (x\times \sigma ,x^{\prime }\times \sigma ^{\prime })=\Phi (x-x^{\prime
})-\phi (x-x^{\prime })\sigma \sigma ^{\prime }  \label{pair-int}
\end{equation}%
and are characterized by a single-particle probability measure $\chi $ on $S$%
. Here $\Phi :\mathbb{R}^{d}\rightarrow \mathbb{R}\cup \{+\infty \}$ and $%
\phi :\mathbb{R}^{d}\rightarrow \mathbb{R}_{+}$ are suitable functions, see
Section \ref{S3} below.

The Gibbs measures of the model are defined as probability measures on the
space $\Gamma (X,S)=\left\{ \hat{\gamma}\subset X\times S:p_{X}(\hat{\gamma}%
)\subset \Gamma (X)\right\} $ of marked configurations, where $p_{X}$ is the
natural projection $X\times S\rightarrow X$ and $\Gamma (X)\ $is the space
of locally finite subsets of $X$. As is typical for systems with unbounded
spins, cf. \cite{DKKP,KKP1,KP,LP}, we work with the Gibbs measures that are
supported on the configurations satisfying certain bounds on their density
and spin growth (the \emph{tempered} Gibbs measures). In the study of the
set $\mathcal{G}^{\mathrm{t}}(\Psi ,\chi )$ of all such measures one
typically poses the following questions:

\begin{itemize}
\item[(E)] \textit{Existence}: is $\mathcal{G}^{\mathrm{t}}(\Psi ,\chi )$
not empty?

\item[(U)] \textit{Uniqueness:} is $\mathcal{G}^{\mathrm{t}}(\Psi ,\chi )$ a
singleton?

\item[(M)] \textit{Multiplicity}: does $\mathcal{G}^{\mathrm{t}}(\Psi ,\chi
) $ contain at least two elements?
\end{itemize}

Usually, only sufficient conditions for positive answer to these questions
are obtained, which justifies distinguishing between (U) and (M). Positive
answer to (M) indicates the appearance of phase transitions in the system.
The comprehensive answer to all the three questions is known only for the
classical Ising model where $X=\mathbb{Z}^{d}$ and $S=\left\{ -1,1\right\} $%
, see e.g., \cite{Simon}. (E) is well-studied also for more general
`crystalline' type spin models, including the case of $X$ being a general
graph and $S=\mathbb{R}$, see \cite{KKP1}. For continuum models with $X=%
\mathbb{R}^{d}$ and compact spin space $S$, (E) is essentially similar to
the case of continuum gas models without spins, see \cite{KPR} and the
references therein. In the case of $S=\mathbb{R}$, questions (E) and (U)
have recently been studied in \cite{DKPP}. In \cite{GeHa}, question (M) was
studied by exploiting a continuum version of the random cluster model and
the percolation theory.

In the present work, we give an answer to question (M) in the general case
of $X=\mathbb{R}^{d}$, $S=\mathbb{R}$ in the absence of the restrictive
`hard core' and `compactness of spins' conditions. Instead, we assume the
strong superstability of the position-position interaction and the
exponential moment bound of the single-particle measure $\chi $, see Section %
\ref{S3}. We exploit the fibre bundle structure of the space $\Gamma (X,S)$
studied in \cite{DKKP,DKKP1}, which allows us to disintegrate any element $%
\mu \in \mathcal{G}^{\mathrm{t}}(\Psi ,V)$ as $\mu (d\hat{\gamma})=\omega
_{\gamma }(d\sigma )\left( p_{X}^{\ast }\mu \right) (d\gamma )$, where $%
\omega _{\gamma }(d\sigma )$ is a Gibbs measure on the product space $%
S^{\gamma }$, for a.a. $\gamma \in \Gamma (X)$. This allows for applying a
suitable modification of methods developed in \cite{DKKP}.

The structure of the paper is as follows. In Section \ref{S2}, we present a
number of facts from the theory of marked configuration spaces. The crucial
one is a fibre bundle structure of such spaces. In Section \ref{S3}, we
describe the model, cf. Assumption (M), and present the main result of this
paper in Theorem \ref{phase}. In Subsection \ref{SS3.3}, we sketch the proof
of the existence of tempered Gibbs measures of our model. The proof of
Theorem \ref{phase} is given in Section \ref{Sphase} and is based on Lemma %
\ref{1lm}, which states that the magnetization in local states can be
uniformly positive. The proof of Lemma \ref{1lm} is in turn based on a
modification of Wells' inequality \cite{W} and the result of \cite{H} that
relates the existence of a ferromagnetic phase of the Ising model on a
general graph to the Bernoulli bond percolation thereon. The existence of
such percolation in our framework is stated in Lemma \ref{th-perc} and
proved in Section \ref{sec-phase}, by extending the general scheme proposed
in \cite{GeHa}. The main idea is to pass to an auxiliary percolation model,
which is dominated by the percolation in question, see Lemmas \ref{5lm}, \ref%
{6lm}, and \ref{lemma1}.

\section{Marked configuration spaces}

\label{S2}

\subsection{The spaces of configurations}

%\label{SS2.2}

The configuration space on $X=\mathbb{R}^{d}$, $d\geq 1$, is
\begin{equation}
\Gamma (X)=\left\{ \gamma \subset X:\ N\left( \gamma _{\Lambda }\right)
<\infty \text{ for any }\Lambda \in \mathcal{B}_{0}(X)\right\} ,
\label{gamma}
\end{equation}%
where $\mathcal{B}_{0}(X)$ is the collection of all compact subsets of $X$, $%
\gamma _{\Lambda }:=\gamma \cap \Lambda $, and $N\left( \cdot \right) $
denotes cardinality. Let $C_{0}(X)$ be the set of all continuous functions $%
f:X\rightarrow \mathbb{R}$ with compact support. The configuration space $%
\Gamma (X)$ is endowed with the vague topology, which is the weakest
topology that makes continuous all the maps
\begin{equation*}
\Gamma (X)\ni \gamma \mapsto \sum_{x\in \gamma }f(x),\quad f\in C_{0}(X).
\end{equation*}%
This topology is metrizable in the way that makes $\Gamma (X)$ a Polish
space (see, e.g., \cite[Section 15.7.7]{Kal} or \cite[Proposition 3.17]{Res}%
). An explicit construction of the appropriate metric can be found in \cite%
{KKut}. By $\mathcal{P}(\Gamma (X))$ we denote the set of all probability
measures on the Borel $\sigma $-algebra $\mathcal{B}(\Gamma (X))$ of subsets
of $\Gamma (X)$.

\begin{remark}
\label{rem-conf}In a similar fashion, the configuration space $\Gamma (Y)$
can be defined for an arbitrary Riemannian manifold $Y$. In Subsection \ref%
{PLm2}, we use the space $\Gamma (X^{(2)})$, where $X^{(2)}\ $ is the
collection of two-element subsets of $X$, which can be identified with the
symmetrization of the space $\left( X\times X\right) \setminus \left\{
(x,x):x\in X\right\} $ and thus possesses a Riemannian manifold structure.
\end{remark}

Let us now consider the product $X\times S$, $S=\mathbb{R}$. The canonical
projection $p_{X}:X\times S\rightarrow X$ can naturally be extended to the
configuration space $\Gamma (X\times S)$. However, for a configuration $\hat{%
\gamma}\in \Gamma (X\times S)$, its image $p_{X}(\hat{\gamma})$ is a subset
of $X$ that in general admits accumulation and multiple points and hence
does not belong to $\Gamma (X)$. The marked configuration space\textit{\ }$%
\Gamma (X,S)$ is defined in the following way:%
\begin{equation*}
\Gamma (X,S)=\left\{ \hat{\gamma}\in \Gamma (X\times S):p_{X}(\hat{\gamma}%
)\in \Gamma (X)\right\} .
\end{equation*}%
The space $\Gamma (X,S)$ is endowed with a metrizable topology defined as
the weakest topology that makes continuous the maps
\begin{equation}
\Gamma (X,S)\ni \hat{\gamma}\mapsto \sum_{x\in p_{X}(\hat{\gamma}%
)}f(x,\sigma _{x})  \label{top}
\end{equation}%
for all bounded continuous functions $f\in X\times S\rightarrow \mathbb{R}$
such that $\mathrm{supp}f(\cdot ,\sigma )\subset \Lambda $, for some $%
\Lambda \in \mathcal{B}_{0}(X)$ and all $\sigma \in S$. This topology has
been used in \cite{AKLU,CG,KunaPhD}. It makes $\Gamma (X,S)$ a Polish space,
cf. \cite[Section 2]{CG}, where a concrete metric is given. We equip $\Gamma
(X,S)$ with the corresponding Borel $\sigma $-algebra $\mathcal{B}(\Gamma
(X,S))$.

Along with $\Gamma (X,S)$ we will also use the spaces $\Gamma (\Lambda ,S),$
$\Lambda \in \mathcal{B}_{0}(X)$, and the space $\Gamma
_{0}(X,S):=\bigcup_{\Lambda \in \mathcal{B}_{0}(X)}\Gamma (\Lambda ,S)$ of
finite marked configurations, endowed with the Borel $\sigma $-algebras $%
\mathcal{B}(\Gamma (\Lambda ,S))$ and $\mathcal{B}\left( \Gamma
_{0}(X,S)\right) $ respectively, which are induced by the Euclidean
structure of $X$. It is known that $\mathcal{B}(\Gamma _{0}(X,S))=\{A\cap
\Gamma _{0}(X,S):A\in \mathcal{B}(\Gamma (X,S))\}$.

The spaces $\Gamma (\Lambda ,S)$ and $\Gamma _{0}(X,S)$ can be identified
with the corresponding subspaces of $\Gamma (X,S)$ via the natural
embedding. Clearly, these subspaces belong to $\mathcal{B}(\Gamma (X,S))$
and $\sigma $-algebras $\mathcal{B}(\Gamma (\Lambda ,S))$ and $\mathcal{B}
\left( \Gamma _{0}(X,S)\right) $ can be considered as sub-algebras of $%
\mathcal{B}\left( \Gamma (X,S)\right) $.

On the other hand, we can introduce the algebras $\mathcal{B}_{\Lambda
}(\Gamma (X,S))$ of sets $C_{B}:=\left\{ \gamma \in \Gamma (X):\gamma
_{\Lambda }\in B\right\} $, $B\in \mathcal{B}(\Gamma (\Lambda ,S))$ and
define the algebra of local (cylinder) sets
\begin{equation}
\mathcal{B}_{\mathrm{loc}}\left( \Gamma (X,S)\right) :=\bigcup_{\Lambda \in
\mathcal{B}_{0}(X)}\mathcal{B}_{\Lambda }\left( \Gamma (X,S)\right) .
\label{Loc.df}
\end{equation}
In a similar way, one introduces the spaces $\Gamma (\Lambda )$, $\Gamma
_{0}(X)$ and the corresponding algebras $\mathcal{B}(\Gamma (\Lambda ))$, $%
\mathcal{B}(\Gamma _{0}(X))$ and $\mathcal{B}_{\mathrm{loc}}\left( \Gamma
(X)\right) $.

It is possible to show that a given $F:\Gamma_0 (X,S) \to \mathbb{R}$ is $%
\mathcal{B}\left( \Gamma _{0}(X,S)\right) $-measurable if and only if, for
each $n\in \mathbb{N}$, there exists a symmetric Borel function $%
F_{n}:(X\times S)^{n}\rightarrow \mathbb{R}$ such that
\begin{equation*}
F(\hat{\gamma})=F_{n}((x_{1},\sigma _{1}),\dots ,(x_{n},\sigma _{n})),\quad
\hat{\gamma}=\{(x_{1},\sigma _{1}),\dots ,(x_{n},\sigma _{n})\}.
\end{equation*}

For the single-spin measure $\chi \in \mathcal{P}(S)$ ($=:$ the space of
probability measures on $S$) and some $z>0$, we introduce the
Lebesgue-Poisson measure $\hat{\lambda}_{z}$ on $\mathcal{B}(\Gamma
_{0}(X,S))$ by the relation
\begin{eqnarray}
&&\int_{\Gamma _{0}(X,S)}F(\hat{\gamma})\hat{\lambda}_{z}(d\hat{\gamma}%
)=F(\emptyset )  \label{LPM} \\[0.08in]
&&+\sum_{n=1}^{\infty }\frac{z^{n}}{n!}\int_{\left( X\times S\right)
^{n}}F_{n}((x_{1},\sigma _{1}),\dots ,(x_{n},\sigma _{n}))\chi (d\sigma
_{1})dx_{1}\cdots \chi (d\sigma _{n})dx_{n},  \notag
\end{eqnarray}
which has to hold for all measurable $F:\Gamma _{0}(X,S)\rightarrow \mathbb{R%
}_{+}$. Likewise, the Lebesgue-Poisson measure $\lambda _{z}$ on $\mathcal{B}
(\Gamma _{0}(X))$ is defined by
\begin{equation}
\int_{\Gamma _{0}(X)}F(\gamma )\lambda _{z}(d\gamma )=F(\emptyset
)+\sum_{n=1}^{\infty }\frac{z^{n}}{n!}\int_{X^{n}}F_{n}(x_{1},\dots
,x_{n})dx_{1}\cdots dx_{n},  \label{LPMb}
\end{equation}
holding for all measurable $F:\Gamma _{0}(X)\rightarrow \mathbb{R}_{+}$.

\subsection{Disintegration of measures}

\label{SS2.1a}

The space $\Gamma (X,S)$ has the structure of a fibre bundle over $\Gamma
(X) $, with fibres $p_{X}^{-1}(\gamma )$ which can be identified with the
product
\begin{equation*}
S^{\gamma }=\prod\limits_{x\in \gamma }S_{x},\qquad S_{x}=S.
\end{equation*}%
Thus, each $\hat{\gamma}\in \Gamma (X,S)$ can be represented by the pair
\begin{equation*}
\hat{\gamma}=(\gamma ,\sigma _{\gamma }),\quad \text{ where }\gamma =p_{X}(%
\hat{\gamma})\in \Gamma (X),\ \ \sigma _{\gamma }=(\sigma _{x})_{x\in \gamma
}\in S^{\gamma }.
\end{equation*}%
It follows directly from the definition of the corresponding topologies that
the map $p_{X}:\Gamma (X,S)\rightarrow \Gamma (X)$ is continuous. For each $%
B\in \mathcal{B}(\Gamma (X))$, its preimage $p_{X}^{-1}(B)$ is in $\mathcal{B%
}(\Gamma (X,S))$. Likewise, $p_{X}^{-1}(B)\in \mathcal{B}(\Gamma _{0}(X,S))$
for each $B\in \mathcal{B}(\Gamma _{0}(X))$. In particular, $%
p_{X}^{-1}(\gamma )=p_{X}^{-1}(\{\gamma \})=S^{\gamma }\in \mathcal{B}%
(\Gamma _{0}(X,S))\subset \mathcal{B}(\Gamma (X,S))$. We equip each $%
S^{\gamma }$ with the product topology and donote by $\mathcal{B}(S^{\gamma
})$ the corresponding Borel $\sigma $-algebra. By Kuratowski's theorem, see
\cite{Par}, it is possible to show that
\begin{equation*}
\mathcal{B}(S^{\gamma })=\{A\cap S^{\gamma }:A\in \mathcal{B}(\Gamma
(X,S))\}.
\end{equation*}%
Then, for each probability measure $\mu $ on $\mathcal{B}(\Gamma (X,S))$,
one can define its projection $p_{X}^{\ast }\mu $ on $\mathcal{B}(\Gamma
(X)) $ by setting
\begin{equation*}
(p_{X}^{\ast }\mu )(B)=\mu (p_{X}^{-1}(B)),\qquad B\in \mathcal{B}(\Gamma
(X)).
\end{equation*}%
This in turn allows one to disintegrate
\begin{equation}
\mu (d\hat{\gamma})=\omega _{\gamma }(d\sigma _{\gamma })(p_{X}^{\ast }\mu
)(d\gamma ),  \label{di}
\end{equation}%
where $\omega _{\gamma }$ is a probability measure on $\mathcal{B}(S^{\gamma
})$ for $p_{X}^{\ast }\mu $-almost all $\gamma \in \Gamma (X)$. Moreover,
for each $B\in \mathcal{B}(S^{\gamma })$, the map $\gamma \mapsto \omega
_{\gamma }(B)$ is $\mathcal{B}(\Gamma (X))$-measurable. A similar
disintegration can be applied to measures on $\mathcal{B}(\Gamma _{0}(X,S))$%
. In particular, for the measures introduced in (\ref{LPM}) and (\ref{LPMb}%
), one has
\begin{equation}
\hat{\lambda}_{z}(d\hat{\gamma})=\chi _{\gamma }(d\sigma _{\gamma })\lambda
_{z}(d\gamma ),\quad \chi _{\gamma }(d\sigma _{\gamma }):=\bigotimes_{x\in
\gamma }\chi (d\sigma _{x}),\quad \gamma \in \Gamma _{0}(X).  \label{di1}
\end{equation}

\subsection{Tempered marked configurations}

\label{SS2.2a}In the sequel, we use the following partition of $X$. For $%
k=(k^{(1)},\dots ,k^{(d)})\in \mathbb{Z}^{d}$ and $l>0$, we set
\begin{equation}
\Xi _{k}:=\left\{ x\in X:\ x^{(i)}\in \left[ l(k^{(i)}-1/2),l(k^{(i)}+1/2)%
\right) \right\} .  \label{Qkl}
\end{equation}%
Given integer $v>2$, we take $w\in \mathbb{N}$ such that
\begin{equation}
w\geq \frac{2(v-1)}{v-2}.  \label{pq}
\end{equation}%
For these $v$ and $w$, we then define, cf. (\ref{gamma}),
\begin{equation}
F(\hat{\gamma})=\left[ N(\gamma )\right] ^{v}+\sum_{x\in \gamma }|\sigma
_{x}|^{w},\quad \hat{\gamma}=(\gamma ,\sigma _{\gamma })\in \Gamma _{0}(X,S),
\label{Fk}
\end{equation}%
and
\begin{equation}
F_{\alpha }(\hat{\gamma})=\sup_{k\in \mathbb{Z}^{d}}F(\hat{\gamma}_{k})\exp
(-\alpha |k|),\quad \hat{\gamma}\in \Gamma (X,S),\quad \alpha >0,
\label{Fka}
\end{equation}%
where $\gamma _{k}:=\gamma \cap \Xi _{k}$. By means of these functions we
then set
\begin{equation}
\Gamma ^{\mathrm{t}}(X,S)=\left\{ \hat{\gamma}\in \Gamma (X,S):\ F_{\alpha }(%
\hat{\gamma})<\infty \text{ for each }\alpha >0\right\} ,  \label{temper0}
\end{equation}%
which is the space of tempered marked configurations. Note that $\Gamma ^{%
\mathrm{t}}(X,S)\in \mathcal{B}(\Gamma (X,S))$ and is independent of $l$
used in (\ref{Qkl}). In a similar way, we can define the space $\Gamma ^{%
\mathrm{t}}(X)$ of tempered configurations in $X$ using the function $%
F_{X}(\gamma ):=[N(\gamma )]^{v}$ in place of $F(\hat{\gamma})$. Observe
that, for any $\gamma \in \Gamma ^{\mathrm{t}}(X)$ and $\sigma _{\gamma
}=(\sigma _{x})_{x\in \gamma }$ with $\sup_{x\in \gamma }\left\vert \sigma
_{x}\right\vert <\infty $, we have $\left( \gamma ,\sigma _{\gamma }\right)
\in \Gamma ^{\mathrm{t}}(X,S)$.

\begin{definition}
\label{TempMdf} A probability measure $\nu$ on $\mathcal{B}(\Gamma (X,S))$
is said to be tempered if $\nu(\Gamma^{\mathrm{t}}(X,S)) =1$.
\end{definition}

\section{The model and main result}

\label{S3}

\subsection{Description of the model}

\label{SS3.1} The interaction between the particles is supposed to be
pair-wise and consisting of position-position and spin-spin parts described
by measurable functions $\Phi :\mathbb{R}^{d}\rightarrow \mathbb{R}\cup
\{+\infty \}$ and $\phi :\mathbb{R}^{d}\rightarrow \mathbb{R}_{+}$,
respectively, cf. (\ref{pair-int}). Another model `parameter' is a
single-spin measure $\chi \in \mathcal{P}(S)$. Since $\phi \geq 0$, the
spin-spin interaction is of ferromagnetic type, cf. (\ref{H}). By $\Phi _{+}$
we denote the positive part of $\Phi $, i.e., $\Phi _{+}=\max \{\Phi ;0\}$.
Thereby, for $\hat{\gamma}=(\gamma ,\sigma _{\gamma })$ with $\gamma \in
\Gamma _{0}(X)$ and $\sigma _{\gamma }\in S^{\gamma }$, we define
\begin{eqnarray}
H(\gamma ) &=&\sum_{\left\{ x,y\right\} \in \gamma }\Phi (x-y),  \label{H} \\%
[.2cm]
E(\sigma _{\gamma }) &=&-\sum_{\left\{ x,y\right\} \in \gamma }\phi
(x-y)\sigma _{x}\sigma _{y}.  \notag
\end{eqnarray}
The model parameters are supposed to satisfy the following

\begin{assumption}[\textit{M}]
$\frac{{}}{{}}$

\begin{enumerate}
\item[(1)] There exists $r>0$ such that $\Phi_{+} (x)=0$ whenever $|x|> r$.

\item[(2)] For each $\delta>0$, there exists $C_\delta < +\infty$ such that
\begin{equation}
\int_{\left\vert x\right\vert \geq \delta }\Phi _{+}( x )dx\leq C_\delta
<\infty.  \label{int}
\end{equation}

\item[(3)] $\Phi $ is bounded from below and there exist $\epsilon >0$ and
positive $A_{\Phi }$, $B_{\Phi }$ such that
\begin{equation}
H(\gamma )\geq A_{\Phi }\sum_{k\in \mathbb{Z}^{d}}[N(\gamma
_{k})]^{v+\epsilon }-B_{\Phi }N(\gamma ),\quad \gamma _{k}=\gamma \cap \Xi
_{k},  \label{sss}
\end{equation}%
for any $\gamma \in \Gamma _{0}(X)$, where $v$ is as in (\ref{Fk}).

\item[(4)] $\phi :\mathbb{R}^{d}\rightarrow \mathbb{R}_{+}$ is bounded and
such that there exist $\phi _{\ast }>0$ and $R>0$, for which the following
holds
\begin{equation}
\phi (x)\geq \phi _{\ast },\ \ \mathrm{for}\ |x|\leq R;\qquad \phi (x)=0,\ \
\mathrm{for}\ |x|>R.  \label{phi-star1}
\end{equation}

\item[(5)] The measure $\chi \in \mathcal{P}(S)$ is symmetric with respect
to $\sigma \rightarrow -\sigma $. There exist constants $\varkappa >0$ and $%
u>w$, see (\ref{pq}), such that
\begin{equation}
\int_{S}\exp \left( \varkappa |s|^{u}\right) \chi (ds)<\infty ,
\label{s-quad}
\end{equation}%
and $\chi (\left\{ 0\right\} )<1$.

\item[(6)] The parameters $r$ and $R$ satisfy the relation
\begin{equation}
r<R/4.  \label{pq-rel}
\end{equation}
\end{enumerate}
\end{assumption}

\begin{remark}
\label{Phirk} Clearly, positive $\epsilon$ in (\ref{sss}) can be chosen in
such a way that $u$ in (\ref{s-quad}) also satisfies $u > 2(v+\epsilon
-1)/(v+\epsilon -2)$, which is important for proving Proposition \ref%
{equicont1} below, see \cite{DKPP}.
\end{remark}

The property of $\Phi $ as in (\ref{sss}) is called \textit{strong
superstability} \cite{Ru70}. One of the best-understood examples of
interaction of this type is given by the potential, which satisfies $\Phi
\left( x\right) \geq c|x|^{-d(1+\epsilon )}$ in the vicinity of $x=0$. In
this case, one can take any $v>2$. For a detailed study and historical
comments see \cite{RT} and also \cite[Remark 4.1.]{KPR}.

\subsection{Main result}

\label{SS3.2} For $\Delta \subset X$, we write $\Delta ^{c}=X\setminus
\Delta $. Given $\Delta \in \mathcal{B}_{0}(X)$, for $\hat{\eta}=(\eta
,\sigma _{\eta })\in \Gamma (\Delta ,S)$ and $\hat{\gamma}=(\gamma ,\sigma
_{\gamma })\in \Gamma (\Delta ^{c},S)$, we set
\begin{equation}
H(\eta |\gamma )=H(\eta )+\sum_{x\in \eta }\sum_{y\in \gamma }\Phi (x-y)
\label{cond-en-H}
\end{equation}%
and%
\begin{equation}
E(\sigma _{\eta }|\sigma _{\gamma })=E(\sigma _{\eta })-\sum_{x\in \eta
}\sum_{y\in \gamma }\phi (x-y)\sigma _{x}\sigma _{y}.  \label{cond-en}
\end{equation}%
The Gibbs specification $\Pi $ of the model is the family of probability
kernels $\Pi _{\Delta }$, $\Delta \in \mathcal{B}_{0}(X)$, defined by the
integrals
\begin{eqnarray}
\int_{\Gamma (X,S)}F(\hat{\eta})\Pi _{\Delta }\left( d\hat{\eta}|\hat{\gamma}%
\right) &=&[Z_{\Delta }(\hat{\gamma})]^{-1}\int_{\Gamma (\Delta ,S)}F(\hat{%
\eta}_{\Delta }\cup \hat{\gamma}_{\Delta ^{c}})  \label{spec} \\[0.2cm]
&\times &\exp \bigg{(}-H(\eta _{\Delta }|\gamma _{\Delta ^{c}})-E(\sigma
_{\eta _{\Delta }}|\sigma _{\gamma _{\Delta ^{c}}})\bigg{)}\hat{\lambda}%
_{z}(d\hat{\eta}_{\Delta }),  \notag
\end{eqnarray}%
which has to hold for all measurable functions $F:\Gamma (X,S)\rightarrow
\mathbb{R}_{+}$ and all $\hat{\gamma}\in \Gamma ^{\mathrm{t}}(X,S)$, see (%
\ref{temper0}). Here $\hat{\lambda}_{z}$ is the marked Lebesgue-Poisson
measure defined in (\ref{LPM}) and
\begin{equation*}
Z_{\Delta }(\hat{\gamma})=\int_{\Gamma (\Delta ,S)}\exp \bigg{(}-H(\eta
_{\Delta }|\gamma _{\Delta ^{c}})-E(\sigma _{\eta _{\Delta }}|\sigma
_{\gamma _{\Delta ^{c}}})\bigg{)}\ \hat{\lambda}_{z}(d\hat{\eta}_{\Delta })
\end{equation*}%
is the normalizing factor (partition function) making $\Pi _{\Delta }\left(
\cdot \left\vert \hat{\gamma}\right. \right) $ a probability measure on $%
\Gamma (X,S)$, provided $Z_{\Delta }(\hat{\eta})\neq 0$ which is the case
under Assumption (M), see \cite{DKPP}.

A probability measure $\nu \in \mathcal{P}(\Gamma(X,S))$ is said to be a
Gibbs measure associated with the specification $\Pi $ if it satisfies the
Dobrushin-Lanford-Ruelle (DLR) equation
\begin{equation}
\nu (B)=\int_{\Gamma (X,S)}\Pi _{\Delta }\left( B\left\vert \hat{\gamma}%
\right. \right) \nu (d\hat{\gamma}) ,  \label{DLR}
\end{equation}%
which has to hold for all $B\in \mathcal{B}(\Gamma (X,S))$ and $\Delta \in
\mathcal{B}_{0}(X)$. By $\mathcal{G}^{\mathrm{t}} (\Gamma(X,S))$ we denote
the set of all tempered Gibbs measures, see Definition \ref{TempMdf}. The
result of this work is given in the following

\begin{theorem}
\label{phase} Let Assumption (M) hold and $d\geq 2$. Then there exists $%
z_{c}>0$ such that
\begin{equation*}
N(\mathcal{G}^{\mathrm{t}}(\Gamma (X,S)))\geq 2
\end{equation*}%
for all $z>z_{c}$.
\end{theorem}

Observe that Theorem \ref{phase} contains two quite different in their
nature statements: (i) $N(\mathcal{G}^{\mathrm{t}}(\Gamma (X,S)))\neq
\emptyset $ and (ii) $\mathcal{G}^{\mathrm{t}}(\Gamma (X,S))$ contains at
least two elements. In the next section, we present a sketch of the proof of
(i). A complete proof of this is given in \cite{DKPP}. The proof of (ii) is
based on the comparison with the classical Ising model on a random geometric
graph and its relationship with percolation theory on this graph and will be
given in Sections \ref{Sphase} and \ref{sec-phase}.

\subsection{The existence of Gibbs measures}

\label{SS3.3}

The main idea here is to show that, for at least some $\hat{\gamma}$, the
family
\begin{equation*}
\{\Pi _{\Lambda }(\cdot |\hat{\gamma})\}_{\Lambda \in \mathcal{B}%
_{0}(X)}\subset \mathcal{P}(\Gamma (X,S))
\end{equation*}%
has accumulation points, which solve (\ref{DLR}) and are tempered measures
in the sense of Definition \ref{TempMdf}. These accumulation points are
sought in the local set convergence topology ($\mathfrak{L}$-topology),
which is defined as the weakest topology on $\mathcal{P}(\Gamma (X,S))$ that
makes continuous all the evaluation maps $\mu \mapsto \mu (A)$, $A\in
\mathcal{B}_{\mathrm{loc}}(\Gamma (X,S))$, see (\ref{Loc.df}). This topology
is weaker than the usual weak topology for which the relative compactness is
established by means of Prokhorov's theorem, see, e.g., \cite{Par}. Instead
we can use the following instruments, cf. \cite[Prop. 4.9]{Geor}.

\begin{definition}
\label{equi} A sequence $\left\{ \mu _{n}\right\} _{n\in \mathbb{N}} \subset
\mathcal{P}(\Gamma (X,S))$ is said to be locally equicontinuous (LEC) if for
any $\Delta \in \mathcal{B}_{0}(X)$ and any $\left\{ B_{m}\right\} _{m\in
\mathbb{N}}\subset \mathcal{B}(\Gamma (\Delta ,S))$, $B_{m}\searrow
\emptyset $, $m\rightarrow \infty $, it follows that
\begin{equation*}
\lim_{m\rightarrow \infty} \limsup_{n\rightarrow \infty} \mu _{n}\left(
B_{m}\right) =0.
\end{equation*}
\end{definition}

\begin{proposition}
\label{equicont1}Each LEC sequence $\left\{ \mu _{n}\right\} _{n\in \mathbb{N%
}} \subset \mathcal{P}(\Gamma (X,S))$ has accumulation points in the $%
\mathfrak{L}$-topology, which are probability measure on $\Gamma (X,S)$.
\end{proposition}

Let $\{\Lambda _{m}\}_{m\in \mathbb{N}}$, be an exhausting sequence of
compact subsets of $X$. This means that $\Lambda_m \subset \Lambda_{m+1}$
for all $m\in \mathbb{N}$ and $\Lambda _{m}\nearrow X,\ m\rightarrow \infty $%
. Set $\Pi_{m}=\Pi _{\Lambda _{m}}\left( \cdot|\hat{\gamma}\right)$, $\hat{%
\gamma} \in \Gamma^{\mathrm{t}}(X,S)$. The following fact was proved in \cite%
{DKPP}.

\begin{proposition}
\label{equicont2} For any $\hat{\gamma}\in \Gamma ^{\mathrm{t}}(X,S)$ and
any choice of the exhausting sequence $\{\Lambda _{m}\}_{m\in \mathbb{N}}$,
the sequence $\left\{ \Pi _{m}\right\} _{m\in \mathbb{N}}$ is LEC.
\end{proposition}

The next theorem states sufficient conditions for the existence and
uniqiness of tempered Gibbs measures.

\begin{theorem}
\label{Atm}\cite{DKPP} Under Assumption (M) the following holds.

\begin{enumerate}
\item[(i)] The set $\mathcal{G}^{\mathrm{t}}(\Gamma (X,S))$ is nonempty;
each of its elements has the property
\begin{equation}
\forall \alpha >0\qquad \quad \underset{k\in \mathbb{Z}^{d}}{\mathrm{sup}}%
\int_{\Gamma (X,S)}e^{\alpha F(\hat{\gamma}_{k})}\mu \left( d\hat{\gamma}%
\right) <\infty ,  \label{moment-est}
\end{equation}%
cf. (\ref{Fk}) and (\ref{Fka}).

\item[(ii)] There exist constants $\phi _{0},z_{0}>0$ such that $N\left(
\mathcal{G}(\Gamma ^{t}(X,S))\right) =1$ whenever $\phi (x)\leq \phi _{0}$, $%
|x|\leq R$, and $z\leq z_{0}$.
\end{enumerate}
\end{theorem}

In order to fix certain notations, we give a sketch of the proof of (i)%
\textit{.} It follows from Proposition \ref{equicont2} that, for any $\hat{%
\gamma}\in \Gamma ^{\mathrm{t}}(X,S)$, the sequence $\left\{ \Pi
_{n}\right\} _{n\in \mathbb{N}}$ has an accumulation point $\mu ^{\hat{\gamma%
}}\in \mathcal{P}(\Gamma (X,S))$, so that there exists a subsequence $%
\Lambda _{n_{j}},\ j\in \mathbb{N}$ such that
\begin{equation}
\mu ^{\hat{\gamma}}(B)=\underset{j\rightarrow \infty }{\mathrm{lim}}~\Pi
_{\Lambda _{n_{j}}}\left( B\left\vert \hat{\gamma}\right. \right) ,
\label{lim-measure}
\end{equation}%
holding for any $B\in \mathcal{B}_{0}(\Gamma (X,S))$. Standard limit
transition arguments show that $\mu ^{\hat{\gamma}}$ satisfies (\ref{DLR})
and the estimate in (\ref{moment-est}).

\section{Proof of the main result}

\label{Sphase}

\subsection{Proof of Theorem \protect\ref{phase}}

From now on we fix the value of $l$ in (\ref{Qkl}) by setting
\begin{equation}
l=R/2\sqrt{d},  \label{lk1}
\end{equation}%
where $R$ is as in (\ref{phi-star1}). Then by $\mathcal{L}\subset \mathcal{B}%
_{0}(X)$ we denote the family of finite unions of the cells defined in (\ref%
{Qkl}) such that $\Xi _{0}$ is contained in each $\Lambda \in \mathcal{L}$.
Next, for $n_{\ast }\in \mathbb{N}$ and $a>0$, we define the sets $\widehat{%
\Gamma }_{\pm }(n_{\ast },a)\subset \Gamma ^{\mathrm{t}}(X,S)$ as consisting
of all those $\hat{\gamma}=(\gamma ,\sigma _{\gamma })$ that satisfy the
following two conditions:
\begin{equation}
(a)\ \forall k\in \mathbb{Z}^{d}\quad N(\gamma _{k})\geq n_{\ast };\qquad
(b)\ \forall x\in \gamma \quad \sigma _{x}=\pm a.  \label{1}
\end{equation}%
In view of (\ref{Fka}) and (\ref{temper0}), each $\hat{\gamma}\in \widehat{%
\Gamma }_{\pm }(n_{\ast },a)$ should have the property: for every $\alpha >0$%
, there exists $N_{\alpha }>0$ such that
\begin{equation*}
\forall k\in \mathbb{Z}^{d}\quad N(\gamma _{k})\leq N_{\alpha }e^{\alpha
|k|},
\end{equation*}%
i.e., $\gamma $ should be in $\Gamma ^{\mathrm{t}}(X)$. Now we set
\begin{equation}
M(\hat{\gamma})=\sum_{x\in \gamma _{0}}\sigma _{x},\qquad \hat{\gamma}\in
\Gamma ^{\mathrm{t}}(X,S).  \label{2}
\end{equation}%
The map $\Gamma (X,S)\ni \hat{\gamma}\mapsto M(\hat{\gamma})$ is clearly
measurable, cf. (\ref{top}). The proof of Theorem \ref{phase} is based on
the following result, which will be gradually proved in the remaining part
of the paper.

\begin{lemma}
\label{1lm} Under the assumptions of Theorem \ref{phase}, there exist $%
z_{c}>0$, $n_{\ast }\in \mathbb{N}$ and positive constants $a$ and $m_{c}$
such
\begin{equation}
\int_{\Gamma (X,S)}M(\hat{\gamma})\Pi _{\Lambda }(d\hat{\gamma}|\hat{\xi}%
)\geq m_{c}.  \label{3}
\end{equation}%
for any $z>z_{c}$, $\hat{\xi}\in \widehat{\Gamma }_{+}(n_{\ast },a)$ and $%
\Lambda \in \mathcal{L}$.
\end{lemma}

Note that $M(\hat{\gamma})$ clearly is $\Pi _{\Lambda }(\cdot |\hat{\xi})$%
-integrable for each $\hat{\xi}\in \Gamma ^{\mathrm{t}}(X,S)$.

\noindent \textit{Proof of Theorem \ref{phase}.} Given $\hat{\xi}\in
\widehat{\Gamma }_{+}(n_{\ast },a)$, let $\hat{\xi}^{-}\in \widehat{\Gamma }%
_{-}(n_{\ast },a)$ be such that $p_{X}(\hat{\xi})=p_{X}(\hat{\xi}^{-})$. By (%
\ref{spec}) we then get
\begin{equation}
\int_{\Gamma (X,S)}M(\hat{\gamma})\Pi _{\Lambda }(d\hat{\gamma}|\hat{\xi}%
)=-\int_{\Gamma (X,S)}M(\hat{\gamma})\Pi _{\Lambda }(d\hat{\gamma}|\hat{\xi}%
^{-}).  \label{4}
\end{equation}%
For $n\in \mathbb{N}$, let $\Lambda _{n}$ be the union of all $\Xi _{k}$
with $|k|\leq n$. For such $\hat{\xi}$ and $\hat{\xi}^{-}$, both sequences $%
\{\Pi _{\Lambda _{n}}(\cdot |\hat{\xi})\}$ and $\{\Pi _{\Lambda _{n}}(\cdot |%
\hat{\xi}^{-})\}$ are relatively compact in the $\mathfrak{L}$-topology.
Thus, one can pick the subsequence $n_{j}$, $j\in \mathbb{N}$, such that the
following holds:
\begin{equation*}
\Pi _{\Lambda _{n_{j}}}(\cdot |\hat{\xi})\rightarrow \mu ^{\hat{\xi}},\qquad
\Pi _{\Lambda _{n_{j}}}(\cdot |\hat{\xi})\rightarrow \mu ^{\hat{\xi}%
^{-}},\quad j\rightarrow +\infty ,
\end{equation*}
see Propositions \ref{equicont1}, \ref{equicont2} and formula (\ref%
{lim-measure}). As in the proof of Theorem \ref{Atm}, his convergence yields
that both $\mu ^{\hat{\xi}}$ and $\mu ^{\hat{\xi}^{-}}$ belong to $\mathcal{G%
}^{\mathrm{t}}(\Gamma (X,S))$. At the same time, by means of (\ref%
{moment-est}), Lemma \ref{1lm} and standard limit transition arguments, we
conclude from (\ref{4}) that $\mu ^{\hat{\xi}}\neq \mu ^{\hat{\xi}^{-}}$,
and the result follows.
%TCIMACRO{\TeXButton{\hfill}{\hfill}}%
%BeginExpansion
\hfill%
%EndExpansion
$\square $

\subsection{Proof of Lemma \protect\ref{1lm}}

\label{PLm1} Given $\Lambda \in \mathcal{B}_0(X)$ and $\hat{\xi}\in \Gamma^{%
\mathrm{t}}(X,S)$, we set $P_\Lambda^{\hat{\xi}} = p^*_X \Pi_\Lambda (\cdot|%
\hat{\xi})$. Then, cf. (\ref{di}) and (\ref{di1}),
\begin{equation}  \label{di2}
\Pi_\Lambda (d\hat{\gamma}|\hat{\xi}) = \pi_{\Lambda, \gamma}^{\hat{\xi}} (d
\sigma_\gamma) P_\Lambda^{\hat{\xi}} ( d \gamma).
\end{equation}
Here
\begin{eqnarray}  \label{di3}
\pi_{\Lambda, \gamma}^{\hat{\xi}} (d \sigma_\gamma) & = & \frac{1}{Q_\Lambda
(\hat{\xi}_{\Lambda^c})} \exp\bigg{(} -
E(\sigma_{\gamma_{\Lambda}}|\sigma_{\xi_{\Lambda^c}})\bigg{)}%
\chi_{\gamma_\Lambda}( d \sigma_{\gamma_{\Lambda}}) \otimes
\delta_{\sigma_{\gamma_{\Lambda^c}}} ( d \sigma_{\gamma_{\Lambda^c}}),
\notag \\[.2cm]
Q_\Lambda (\hat{\xi}_{\Lambda^c}) & = & \int_{S^{\gamma_\Lambda}} \exp%
\bigg{(} - E(\sigma_{\gamma_{\Lambda}}|\sigma_{\xi_{\Lambda^c}})\bigg{)}%
\chi_{\gamma_\Lambda}( d \sigma_{\gamma_{\Lambda}}),
\end{eqnarray}
and
\begin{equation}  \label{di3A}
P_\Lambda^{\hat{\xi}} ( d \gamma) = \frac{Q_\Lambda (\hat{\xi}_{\Lambda^c})}{%
Z_\Lambda(\hat{\xi})} \exp\bigg{(}-H(\gamma_\Lambda|\xi_{\Lambda^c}) \bigg{)}%
\lambda_z (d \gamma_\Lambda)\otimes \delta_{\xi_{\Lambda^c}}(d
\gamma_{\Lambda^c}),
\end{equation}
where $\delta_{\cdot}$ is the corresponding Dirac measure.

Among all those $\chi $ that satisfy (\ref{s-quad}) we distinguish the
measure $\chi ^{a}(d\sigma )=[\delta _{-a}(d\sigma )+\delta _{a}(d\sigma
)]/2 $, $a>0$. This choice corresponds to the Ising model with rescaled
spins. It will be used as a reference system. Let $\pi _{\Lambda ,\gamma
}^{a,\hat{\xi}}$ be as in (\ref{di3}) with this $\chi ^{a}$ on the
right-hand side. Next, we let $\tilde{\phi}(x)=\phi _{\ast }I_{R}(x)$, where
$I_{R}$ is the indicator of the ball $B_{R}=\{x\in X:|x|\leq R\}$ and $\phi
_{\ast }$ is as in (\ref{phi-star1}). Finally, by $\tilde{\pi}_{\Lambda
,\gamma }^{a,\hat{\xi}}$ we denote the measure as in (\ref{di3}) with $\chi
^{a}$ and with $\phi $ replaced in (\ref{H}) and (\ref{cond-en}) by $\tilde{%
\phi}$.

The proof of (\ref{3}) is based on the following statement which we prove in
the next section.

\begin{lemma}
\label{2lm} For any $a>0$ there exist $n_{\ast }\in \mathbb{N}$, $z_{c}>0$,
a constant $\theta \in (0,1/2)$ and a family of sets $\Gamma _{\Lambda }(%
\hat{\xi})\in \mathcal{B}(\Gamma (X))$, $\Lambda \in \mathcal{L}$, $\hat{\xi}%
\in \widehat{\Gamma }(n_{\ast },a)$, with the property
\begin{equation}
P_{\Lambda }^{\hat{\xi}}(\Gamma _{\Lambda }(\hat{\xi}))\geq \theta
\label{di3a}
\end{equation}%
and such that $\gamma _{0}\neq \emptyset $ and%
\begin{equation}
\tilde{\pi}_{\Lambda ,\gamma }^{a,\hat{\xi}}(\sigma _{x}=a)\geq \frac{%
1+\theta }{2},\ x\in \gamma _{0},  \label{di4}
\end{equation}%
for all $\gamma \in \Gamma _{\Lambda }(\hat{\xi})$ and $z>z_{c}$.
\end{lemma}

\noindent \textit{Proof of Lemma \ref{1lm}.} By Lemma \ref{2lm} it follows
that $\gamma _{0}\neq \emptyset $ for each $\gamma \in \Gamma _{\Lambda }(%
\hat{\xi})$. For an arbitrary such $\gamma $, we have:
\begin{equation}
\int_{S^{\gamma }}\left( \sum_{x\in \gamma _{0}}\sigma _{x}\right) \pi
_{\Lambda ,\gamma }^{a,\hat{\xi}}(d\sigma _{\gamma })\geq \int_{S^{\gamma
}}\left( \sum_{x\in \gamma _{0}}\sigma _{x}\right) \tilde{\pi}_{\Lambda
,\gamma }^{a,\hat{\xi}}(d\sigma _{\gamma }),  \label{di5}
\end{equation}%
following by the GKS inequalities, see \cite{Simon}. Now we pass to
unbounded spins and take any $\chi $, which is symmetric and satisfies (\ref%
{s-quad}). For this $\chi $ we pick $a>0$ such that
\begin{equation*}
\chi ([a\sqrt{2},+\infty ))\geq \chi ([0,a]).
\end{equation*}%
By Wells' inequality \cite{W}, for this $a$ we have
\begin{equation}
\int_{S^{\gamma }}\left( \sum_{x\in \gamma _{0}}\sigma _{x}\right) \pi
_{\Lambda ,\gamma }^{\hat{\xi}}(d\sigma _{\gamma })\geq \int_{S^{\gamma
}}\left( \sum_{x\in \gamma _{0}}\sigma _{x}\right) \pi _{\Lambda ,\gamma
}^{a,\hat{\xi}}(d\sigma _{\gamma })\geq a\theta ,  \label{di6}
\end{equation}%
see \cite{DKKP} for more detail. The latter estimate in (\ref{di6}) follows
by (\ref{di4}) and (\ref{di5}). Now by (\ref{di2}) we integrate the
left-hand side of (\ref{di6}), take into account (\ref{2}) and (\ref{di3a}),
and obtain (\ref{3}) with $m_{c}=a\theta ^{2}/2$.
%TCIMACRO{\TeXButton{\hfill}{\hfill}}%
%BeginExpansion
\hfill%
%EndExpansion
$\square $

\subsection{Proof of Lemma \protect\ref{2lm}}

\label{PLm2}

The asymmetry stated in (\ref{di4}) can be established by using its
relationship to the Bernoulli bond percolation in the random geometric graph
$(\gamma )_{R}$, which we introduce now. Given a configuration $\gamma \in
\Gamma _{0}(X)$, the vertex set of the graph is set to be $\gamma $. The
edge set is then defined by setting the adjacency relation: $x\sim y$
whenever $|x-y|\leq R$. That is, $(\gamma )_{R}=(\gamma ,\varepsilon
_{\gamma })$, $\varepsilon _{\gamma }=\{\{x,y\}\subset \gamma :\left\vert
x-y\right\vert \leq R\}$. The corresponding probability distribution is
introduced as follows, see \cite{DKKP}. Let $X^{(2)}\ $ be the space of
two-element subsets of $X$ and $E:=\Gamma (X^{(2)})$ (cf. Remark \ref%
{rem-conf}), so that $\varepsilon _{\gamma }\in E$ for any $\gamma \in
\Gamma (X)$. Each $\varpi \in \mathcal{P}(E)$ can be characterized by its
Laplace transform
\begin{equation*}
L_{\varpi }(\kappa ):=\int_{E}\exp \left[ \sum_{\{x,y\}\in \varepsilon }%
\mathrm{log}\left( 1+\kappa (x,y)\right) \right] \varpi (d\varepsilon ),
\end{equation*}%
where $\kappa $ runs over the set $\mathcal{K}$ of all measurable symmetric
functions $X\times X\rightarrow \left( -1,0\right] $. For a given $\gamma
\in \Gamma (X)$, let $\varpi _{\gamma }\in \mathcal{P}(E)$ be the Dirac
measure concentrated at $\varepsilon _{\gamma }$. Its Laplace transform is
then
\begin{equation*}
L_{\varpi _{\gamma }}(\kappa )=\exp \left[ \sum_{\{x,y\}\in \varepsilon
_{\gamma }}\mathrm{log}\left( 1+I_{R}(x-y)\kappa (x,y)\right) \right] ,
\end{equation*}%
where, as above, $I_{R}$ is the indicator of the ball $B_{R}$. For a given $%
q\in \lbrack 0,1]$, the independent \textit{$q$-thinning} of $\varpi
_{\gamma }$ is the measure $\varpi _{\gamma }^{q}\in \mathcal{P}(E)$, cf.
\cite[Section 11.2]{DV}, defined by the relation
\begin{equation}
L_{^{\varpi _{\gamma }^{q}}}(\kappa )=L_{\varpi _{\gamma }}(q\kappa ).
\label{Laplacee}
\end{equation}%
Note that $q\kappa \in \mathcal{K}$. The interpretation of this is that each
$\{x,y\}\in \varepsilon $ is removed from the edge configuration with
probability $1-q$ and is kept with probability $q$. The probability
distribution of such `thinned' configurations is then $\varpi _{\gamma }^{q}$%
. Now let $\Lambda $ and $\hat{\xi}$ be as in the statement of Lemma \ref%
{2lm}, and then $P_{\Lambda }^{\hat{\xi}}$ be as in (\ref{di2}) and (\ref%
{di3a}). For $\varpi _{\gamma }$ and $\varpi _{\gamma }^{q}$ as in (\ref%
{Laplacee}), we define
\begin{equation}
\zeta (d\gamma ,d\varepsilon ):=\varpi _{\gamma }(d\varepsilon )P_{\Lambda
}^{\hat{\xi}}(d\gamma ),\quad \zeta ^{q}(d\gamma ,d\varepsilon ):=\varpi
_{\gamma }^{q}(d\varepsilon )P_{\Lambda }^{\hat{\xi}}(d\gamma ).
\label{q-thinning}
\end{equation}

Let $x\leftrightarrow \infty $ denote the event that $x\in \mathcal{\gamma }$
belongs to an infinite connected component of $(\gamma ,\varepsilon _{\gamma
})$. The proof of Lemma \ref{2lm} is based on the following result proved in
Section \ref{sec-phase}.

\begin{lemma}
\label{th-perc} For any $q\in (0,1)$ and $a>0$ there exist $z_{c}>0$ and $%
n_{\ast }\in \mathbb{N}$ such that the bound
\begin{equation}
\zeta ^{q}\left( \left\{ (\gamma ,\varepsilon ):x\leftrightarrow \infty \
\text{for all }x\in \gamma _{0}\right\} \right) \geq 2\theta  \label{perc0}
\end{equation}%
holds for all $z>z_{c}$, $\Lambda \in \mathcal{L}$, $\hat{\xi}\in \widehat{%
\Gamma }(n_{\ast },a)$ and some constant $\theta \in (0,1/2)$, which depends
only on the dimension of $X$.
\end{lemma}

\noindent \textit{Proof of Lemma \ref{2lm}.} Choose $q$ and $a$ such that
\begin{equation*}
\phi _{\ast }>\frac{a^{2}}{2}\log \frac{1+q}{1-q},
\end{equation*}%
and let $\theta $, $z_{c}$ and $n_{\ast }$ be as in Lemma \ref{th-perc}. Fix
arbitrary $\Lambda \in \mathcal{L}$ and $\hat{\xi}\in \widehat{\Gamma }%
(n_{\ast },a)$. Next, for a given $\gamma \in \Gamma (X)$, set
\begin{equation*}
\Psi (\gamma )=\varpi _{\gamma }^{q}\left( \left\{ \varepsilon
:x\leftrightarrow \infty \ \text{for all }x\in \gamma _{0}\right\} \right) .
\end{equation*}%
Define $\Gamma _{\Lambda }(\hat{\xi})=\{\gamma \in \Gamma (X):\Psi (\gamma
)\geq \theta \}$, where $\theta $ is as in (\ref{perc0}). Since $\Psi
(\gamma )\leq 1$, it follows from (\ref{perc0}) that $P_{\Lambda }^{\hat{\xi}%
}(\Gamma _{\Lambda }(\hat{\xi}))\geq \theta $, hence (\ref{di3a}) holds, and
\begin{equation*}
\varpi _{\gamma }^{q}\left( \left\{ \varepsilon :x\leftrightarrow \infty \
\text{for all }x\in \gamma _{0}\right\} \right) \geq \theta ,\quad \gamma
\in \Gamma _{\Lambda }(\hat{\xi}).
\end{equation*}%
Then (\ref{di4}) follows by \cite[Lemma 4.2]{H}.
%TCIMACRO{\TeXButton{\hfill}{\hfill}}%
%BeginExpansion
\hfill%
%EndExpansion
$\square $

\section{Existence of the percolation}

\label{sec-phase}

Let $\mathsf{Z}=(\mathsf{V},\mathsf{E})$ be the graph with vertex set $%
\mathbb{Z}^{d}$ and the adjacency relation: $k_{1}\sim k_{2}$ whenever $%
|k_{1}-k_{2}|=1$. The main idea of the proof of Lemma \ref{th-perc} is to
construct an auxiliary model on $\mathsf{Z}$ such that the percolation
therein implies (\ref{perc0}).

\subsection{The auxiliary percolation model}

\label{SSq}In this subsection, we fix $\Lambda \in \mathcal{L}$, $n_{\ast
}\in \mathbb{N}$, $a>0$, and $\hat{\xi}\in \widehat{\Gamma }(n_{\ast },a)$.

By $\mathsf{L}\subset \mathsf{V}$ we denote the set of all those $k$ for
which $\Xi _{k}\subset \Lambda $. Next, we introduce two systems of random
variables associated with the graph $(\gamma )_{R}$. Let $\vartheta _{k}$
take value $1$ if the subgraph of $(\gamma )_{R}$ generated by $\gamma _{k}$
is connected and $N(\gamma _{k})\geq n_{\ast }$, and take value 0 otherwise.
For $k_{1}\sim k_{2}$, let $\varsigma _{k_{1}k_{2}}$ take value $1$ if there
exist $x\in \gamma _{k_{1}}$ and $y\in \gamma _{k_{2}}$ such that $x\sim y$
in $(\gamma )_{R}$, and take value $0$ otherwise. Clearly, the maps $(\gamma
,\varepsilon )\mapsto \vartheta _{k}(\gamma ,\varepsilon )$ and $(\gamma
,\varepsilon )\mapsto \varsigma _{k_{1}k_{2}}(\gamma ,\varepsilon )$ are
measurable. In view of the choice of $l$ in (\ref{lk1}), see also (\ref{Qkl}%
), the subgraph of $(\gamma )_{R}$ generated by each $\gamma _{k}$ is
complete; hence, the value of $\vartheta _{k}$ depends only on $N(\gamma
_{k})$. Also due to the choice of $l$, each vertex of $\gamma _{k_{1}}$ is
adjacent (in $(\gamma )_{R}$) to each vertex of $\gamma _{k_{2}}$ whenever $%
k_{1}\sim k_{2}$.

Let $P$ be the joint probability distribution of the random fields $\left\{
\vartheta _{k}\right\}_{k\in \mathsf{V}} $ and $\left\{ \varsigma
_{k_{1}k_{2}}\right\}_{\{k_1,k_2\}\in \mathsf{E}} $ induced by the measure $%
\zeta $ in (\ref{q-thinning}). By the very definition of the set $\widehat{%
\Gamma }(n_{\ast },a)$, see (\ref{1}), we have that $P(\vartheta _{k}=1)=1$
for each $k\in \mathsf{L}^{c}:=\mathsf{V} \setminus \mathsf{L}$, and also $%
P(\varsigma_{k_1k_2}=1)=1$ for all $k_1 \sim k_2$ such that $\vartheta_{k_1}
= \vartheta_{k_2} = 1$. Let $Q$ be the probability measure on $\{0,1\}^{%
\mathsf{V}}\times \{0,1\}^{\mathsf{E}}$ defined as follows. Its projection
on $\{0,1\}^{\mathsf{V}}$ is the product measure such that $Q(\vartheta
_{k}=1)=q_0$ for some $q_0\in (0,1)$ which will be chosen later, and $%
Q(\varsigma _{k_{1}k_{2}}=1)=1$ for all $k_{1}\sim k_{2}$ such that $%
\vartheta_{k_1} = \vartheta_{k_2} = 1$.

As in \cite[Section 3.4]{H}, we introduce the usual componentwise partial
order on $\{0,1\}^{\mathsf{V}}\times \{0,1\}^{\mathsf{E}}$, and the
corresponding increasing real-valued functions on this set. Let $P_{1}$ and $%
P_{2}$ be probability measures on $\{0,1\}^{\mathsf{V}}\times \{0,1\}^{%
\mathsf{E}}$. We say that $P_{2}$ stochastically dominates $P_{1}$ and write
$P_{1}\prec P_{2}$ if
\begin{equation*}
\int fdP_{1}\leq \int fdP_{2}
\end{equation*}%
for each increasing $f$.

We begin by comparing measures $Q$ and $P$ introduced above. Since $%
P(\varsigma _{k_{1}k_{2}}=1)=Q(\varsigma _{k_{1}k_{2}}=1)$ for each $%
k_{1}\sim k_{2}$, we restrict our attention to the random variables $%
\vartheta _{k}$. As in the proof of \cite[Theorem 2.1]{RT1}, by (\ref{spec})
and (\ref{di3A}) one can show that, see also (\ref{NN0}) below,
\begin{equation}
P\left( \vartheta _{k}=1\ \forall k\in \mathsf{V}_{1};\ \ \vartheta _{k}=0\
\forall k\in \mathsf{V}_{2}\right) >0,  \label{irr}
\end{equation}%
which holds for all disjoint $\mathsf{V}_{1},\mathsf{V}_{2}\subset \mathsf{L}
$. Thus, $P$ is irreducible in the sense of \cite[Section 3.4]{H}. Recall
that $P$ depends on the choice of $z$ and $n_{\ast }$, and $Q$ depends on
the choice of $q_{0}\in (0,1)$.

To prove Lemma \ref{th-perc} we need the following result which will be
proved in the next section.

\begin{lemma}
\label{5lm} For each $n_{\ast }\in \mathbb{N}$ and $q_{0}\in (0,1)$ there
exists $z_{c}>0$ such that $Q\prec P$ for any $z>z_{c}$.
\end{lemma}

For a given $q\in (0,1)$ and $n\in \mathbb{N}$, consider an $n$-element set
and connect any two elements of it by an edge with probability $q$,
independently of other edges. Denote by $\varphi (n,q)$ the probability that
the resulting graph is connected. It is known that
\begin{equation}
\varphi (n,q)\geq 1-(n-1)(1-q^{2})^{n-2},\quad n\geq 3,  \label{di10}
\end{equation}%
and hence $\varphi (n,q)\rightarrow 1$ as $n\rightarrow +\infty $, see \cite[%
Lemma 3.4]{GeHa}. By (\ref{di10}) one gets
\begin{equation}
\varrho (n,q):=\inf_{m\geq n}\varphi (m,q)\rightarrow 1\quad \mathrm{as}\ \
n\rightarrow +\infty .  \label{di10a}
\end{equation}%
Likewise, for two sets $A$ and $B$ consisting of $n_{1}$ and $n_{2}$
elements respectively, connect any $a\in A$ and $b\in B$ with each other by
an edge with probaility $q$, independently of other edges. Let $\psi
(n_{1},n_{2},q)$ be the probability that there is at least one edge
connecting $A$ and $B$. Obviously,
\begin{equation}
\psi (n_{1},n_{2},q)=1-(1-q)^{n_{1}n_{2}}.  \label{NN2}
\end{equation}%
Set
\begin{equation}
h(n,q)=\varrho (n,q)\psi (n,n,q).  \label{NK}
\end{equation}

\noindent \textit{Proof of Lemma \ref{th-perc}.} For given $q_{1},q_{2} \in
(0,1)$, let $Q_{q_{1},q_{2}}$ be the measure on $\{0,1\}^{\mathsf{V}}\times
\{0,1\}^{\mathsf{E}}$ such that its projection on $\{0,1\}^{\mathsf{V}}$ is
the product measure for which $Q_{q_{1},q_{2}}(\vartheta _{k}=1)=q_0 q_{1}$,
and $Q_{q_{1},q_{2}}(\varsigma _{k_{1}k_{2}}=1)=q_{2}$ for all $k_{1}\sim
k_{2}$ such that $\vartheta_{k_1} = \vartheta_{k_2} = 1$. That is, $%
Q_{q_{1},q_{2}}$ is the corresponding thinning of the measure $Q$.

For a finite $\mathsf{V}^{\prime }\subset \mathsf{V}$, let $\mathsf{G}%
^{\prime }:=(\mathsf{V}^{\prime },\mathsf{E}^{\prime })$ be a subgraph of $%
\mathsf{Z}$. By $|\mathsf{V}^{\prime }|$ and $|\mathsf{E}^{\prime }|$ we
denote the cardinalities of the corresponding sets. Consider the event $A_{%
\mathsf{G}^{\prime }}=\left\{ \vartheta _{k}=1,\ k\in \mathsf{V}^{\prime },\
\text{and~}\varsigma _{k_{1}k_{2}}=1,\ \{k_{1},k_{2}\}\in \mathsf{E}^{\prime
}\right\} $. By Lemma \ref{5lm}, for the corresponding values of the
parameters $n_{\ast },q,z$ and $q_{0}$ we have
\begin{equation}
Q_{q_{1},q_{2}}(A_{\mathsf{G}^{\prime }})=(q_{0}q_{1})^{|\mathsf{V}^{\prime
}|}q_{2}^{|\mathsf{E}^{\prime }|}=q_{1}^{|\mathsf{V}^{\prime }|}q_{2}^{|%
\mathsf{E}^{\prime }|}Q(A_{\mathsf{G}^{\prime }})\leq q_{1}^{|\mathsf{V}%
^{\prime }|}q_{2}^{|\mathsf{E}^{\prime }|}P(A_{\mathsf{G}^{\prime }}).
\label{NN0}
\end{equation}%
The right-hand side can be estimated in terms of the measure $\zeta ^{q}$
defined in (\ref{Laplacee}) and (\ref{q-thinning}). To this end, we set
\begin{equation}
q_{1}=\varphi (n_{\ast },q),\quad q_{2}=\psi (n_{\ast },n_{\ast },q),
\label{NM1}
\end{equation}%
where $\varphi $ is as in (\ref{di10}), (\ref{di10a}). We then have
\begin{eqnarray}
&&q_{1}^{|\mathsf{V}^{\prime }|}q_{2}^{|\mathsf{E}^{\prime }|}P(A_{\mathsf{G}%
^{\prime }})\leq \int_{\Gamma (X)\times E}\left( \prod_{k\in \mathsf{V}%
^{\prime }}\vartheta (\gamma ,\varepsilon )\varphi (N(\gamma _{k}),q)\right)
\label{NM} \\[0.2cm]
&&\quad \times \left( \prod_{\{k_{1},k_{2}\}\in \mathsf{E}^{\prime
}}\varsigma _{k_{1}k_{2}}(\gamma ,\varepsilon )\psi (N(\gamma
_{k_{1}}),N(\gamma _{k_{2}}),q)\right) \zeta (d\gamma ,d\varepsilon )  \notag
\\[0.2cm]
&&\quad =\int_{\Gamma (X)\times E}\left( \prod_{k\in \mathsf{V}^{\prime
}}\vartheta (\gamma ,\varepsilon )\prod_{\{k_{1},k_{2}\}\in \mathsf{E}%
^{\prime }}\varsigma _{k_{1}k_{2}}(\gamma ,\varepsilon )\right) \zeta
^{q}(d\gamma ,d\varepsilon )  \notag \\[0.2cm]
&=&P^{q}(A_{\mathsf{G}^{\prime }}),  \notag
\end{eqnarray}%
where $P^{q}$ is the joint probability distribution of $\left\{ \vartheta
_{k}\right\} _{k\in \mathsf{V}}$ and $\left\{ \varsigma
_{k_{1}k_{2}}\right\} _{\{k_{1},k_{2}\}\in \mathsf{E}}$ induced by the
measure $\zeta ^{q}$ in (\ref{q-thinning}). Combining (\ref{NN0}) and (\ref%
{NM}) we then get $Q_{q_{1},q_{2}}\prec P^{q}$.

Let $0\leftrightarrow \infty $ denote the event that $0\in \mathsf{Z}$
belongs to an infinite connected component of the graph. Then by (\ref{NN0})
and (\ref{NM}), for $q_{1}$ and $q_{2}$ as in (\ref{NM1}) we have
\begin{eqnarray}
Q_{q_{1},q_{2}}(0\leftrightarrow \infty ) &\leq &P^{q}(0\leftrightarrow
\infty )  \label{NM2} \\[0.2cm]
&=&\zeta ^{q}\left( \left\{ (\gamma ,\varepsilon ):x\leftrightarrow \infty \
\text{for all }x\in \gamma _{0}\right\} \right) .  \notag
\end{eqnarray}%
cf. (\ref{perc0}). To estimate the left-hand side of (\ref{NM2}) we proceed
as follows. For a given subgraph $\mathsf{G}\subseteq \mathsf{Z}$, let $%
\theta ^{\mathrm{site}}(p;\mathsf{G})$ (resp. $\theta ^{\mathrm{bond}}(p;%
\mathsf{G})$), $p\in (0,1)$, be the probability of the event $%
0\leftrightarrow \infty $ in the Bernoulli site (resp. bond) percolation
model on $\mathsf{G}$ with site (resp. bond) probability $p$. It is known
that, see \cite{GS},
\begin{equation}
\theta ^{\mathrm{site}}(p;\mathsf{G})\leq p\theta ^{\mathrm{bond}}(p;\mathsf{%
G})\leq \theta ^{\mathrm{bond}}(p;\mathsf{G}).  \label{NN}
\end{equation}%
Let $\mathsf{G}_{p}$ be the random graph obtained from $\mathsf{Z}$ by
independent deleting sites with probability $1-p$. By construction of the
measure $Q_{q_{1},q_{2}}$ and in view of (\ref{NN}) we have the estimate
\begin{eqnarray}
&&Q_{q_{1},q_{2}}(0\leftrightarrow \infty )=\theta ^{\mathrm{bond}}(q_{2};%
\mathsf{G}_{q_{0}q_{1}})\theta ^{\mathrm{site}}(q_{0}q_{1};\mathsf{Z})
\label{di9} \\[0.2cm]
&&\quad \geq \theta ^{\mathrm{site}}(q_{2};\mathsf{G}_{q_{0}q_{1}})\theta ^{%
\mathrm{site}}(q_{0}q_{1};\mathsf{Z})=\theta ^{\mathrm{site}%
}(q_{0}q_{1}q_{2};\mathsf{Z})>0.  \notag
\end{eqnarray}%
For $d\geq 2$, the latter estimate holds whenever
\begin{equation}
q_{0}q_{1}q_{2}>p^{\mathrm{site}}(d),  \label{NN1}
\end{equation}%
where $p^{\mathrm{site}}(d)$ is the threshold probability for the Bernoulli
site percolation on $\mathsf{Z}$. Thus (\ref{di9}) turns into the following
condition, see (\ref{di10a}), (\ref{NK}), (\ref{NN2}), and (\ref{NM1}):
\begin{equation}
q_{0}h(n_{\ast },q)>p^{\mathrm{site}}(d).  \label{A}
\end{equation}

Now we can finalize the proof of Lemma \ref{th-perc}. Fix an arbitrary $q\in
(0,1)$, pick $n_{\ast }$ such that $h(n_{\ast },q)>p^{\mathrm{site}}(d)$ and
choose any $q_{0}<1\ $satisfying (\ref{A}). For these $n_{\ast }$ and $q_{0}$
let $z_{c}$ be as in Lemma \ref{5lm}. Then for any $z>z_{c}$ we have $Q\prec
P$, which yields (\ref{NM2}). Bound (\ref{perc0}) follows now by (\ref{di9})
with $\theta =p^{\mathrm{site}}(d)/2$. Note that estimates (\ref{NM2}) and (%
\ref{di9}) are uniform in $\Lambda $ and $a$, which completes the proof.
%TCIMACRO{\TeXButton{\hfill}{\hfill}}%
%BeginExpansion
\hfill%
%EndExpansion
$\square $

\subsection{Proof of Lemma \protect\ref{5lm}}

We start with the following technical estimate. Recall that the parameters $%
r $ and $R$ satisfy (\ref{pq-rel}). Let $\Xi $ be any of the cells (\ref{Qkl}%
), (\ref{lk1}) and $\Delta \subset \Xi $ be such that $|x-y|>r$ for each $%
x\in \Delta $ and $y\in \Xi ^{c}$. That is, $\Delta =\Xi \setminus \ \{%
\mathrm{boundary}\ \mathrm{layer}\ \mathrm{of}\ \mathrm{thickness}\ r\}$.
Thus, there is no repulsion between the particles located at $x\in \Delta $
and $y\in \Xi ^{c}$. Observe that the Euclidean volume $\mathrm{Vol}(\Delta
) $ is positive in view of (\ref{pq-rel}). Then, for $x\in \Delta $ and $%
\hat{\gamma}\in \Gamma (X,S)$, we set
\begin{equation}
g(\hat{\gamma})=\int_{X}\exp \left( -\sum_{y\in \gamma }\Phi (x-y)\right)
G(x,\hat{\gamma})dx,  \label{g-fun}
\end{equation}%
where
\begin{equation}
G(x,\hat{\gamma})=\int_{S}\exp \left( s\sum_{y\in \gamma }\phi (x-y)\sigma
_{y}\right) \chi (ds).  \label{dib1}
\end{equation}

\begin{lemma}
\label{6lm} For an arbitrary $n_{\ast }\in \mathbb{N}$, there exists $%
g_{\ast }>0$ such that
\begin{equation}
g(\hat{\gamma})\geq g_{\ast }  \label{est3}
\end{equation}%
for all $\hat{\gamma}\in \Gamma (X,S)$ with $N(\gamma_{\Xi})<n_{\ast }$.
\end{lemma}

\begin{proof}
Fix $n_{\ast }\in \mathbb{N}$ and $\hat{\gamma}=(\gamma ,\sigma_\gamma )$
such that $N(\gamma_{\Xi})<n_{\ast }$. Choose $\delta $ such that $\mathrm{%
Vol~}\Delta -\left( n_{\ast }-1\right) \mathrm{Vol}~(B_{\delta })>0$, where $%
B_{\delta }$ is the ball of radius $\delta $ centered at the origin in $X$.
Define the set $\Delta _{\gamma }$ by removing from $\Delta $ the balls of
radius $\delta $ with centers at the elements of $\gamma \in \Gamma (X)$,
that is,
\begin{equation*}
\Delta _{\gamma }:=\left\{ x\in \Delta :\ \left\vert x-y\right\vert \geq
\delta ,~y\in \gamma \right\} .
\end{equation*}
Then $\mathrm{Vol}(\Delta _{\gamma })\geq \mathrm{Vol}(\Delta )-N(\gamma
_{\Xi})\mathrm{Vol}~(B_{\delta })\geq \mathrm{Vol}(\Delta )-\left( n_{\ast
}-1\right) \mathrm{Vol}~(B_{\delta })=:v_{\ast }$. For a given $c>0$,
introduce the sets
\begin{equation*}
\Delta _{\gamma ,c}:=\left\{ x\in \Delta _{\gamma }:\sum_{y\in \gamma }\Phi
(x-y)\geq c\right\}
\end{equation*}%
and
\begin{equation*}
S_{x,\hat{\gamma}}:=\left\{ s\in S:s\sum_{y\in \gamma }\phi ( x-y)\sigma
_{y}\geq 0\right\} .
\end{equation*}%
For each $x$ and $\hat{\gamma}$, we have either $S_{x,\hat{\gamma}}=\mathbb{R%
}_{\pm}$ or $S_{x,\hat{\gamma}}=\mathbb{R}$, which together with the
symmetry of $\chi $ implies that $\chi (S_{x,\hat{\gamma}})\geq \frac{1}{2}$%
, and hence, see (\ref{dib1})
\begin{equation*}
G(x,\hat{\gamma}) \geq \frac{1}{2} , \qquad x\in X, \ \ \hat{\gamma} \in
\Gamma(X,S).
\end{equation*}
Now we take this into account in (\ref{g-fun}) and obtain
\begin{equation*}
g(\hat{\gamma}) \geq \frac{1}{2} \int_{\Delta} \exp\left(-\sum_{y\in
\gamma}\Phi_{+} (x-y) \right) d y \geq \frac{e^{-c}}{2} \mathrm{Vol}(\Delta
_{\gamma }\diagdown \Delta _{\gamma ,c}).
\end{equation*}
To estimate the latter quantity we use Markov's inequality
\begin{eqnarray*}
\mathrm{Vol}(\Delta _{\gamma ,c}) \leq \frac{1}{c}\int_{\Delta _{\gamma
}}\sum_{y\in \gamma_\Xi}\Phi _{+}(x-y)dx \leq \frac{1}{c}N(\gamma
_{\Xi})\int_{\left\vert x\right\vert >\delta }\Phi _{+}(x)dx,
\end{eqnarray*}
which yields, see (\ref{int}),
\begin{equation*}
\mathrm{Vol}(\Delta _{\gamma }\diagdown \Delta _{\gamma ,c})\geq v_{\ast }-%
\frac{n_{\ast }-1}{c}C_\delta,
\end{equation*}%
so that
\begin{equation*}
g(\hat{\gamma})\geq \frac{1}{2}e^{-c}\left( v_{\ast }-\frac{n_{\ast }-1}{c}%
C_\delta\right) .
\end{equation*}%
It is clear that the right-hand side is positive for sufficiently large $c$.
The (\ref{est3}) follows with $g_{\ast }=\underset{c> 0}{\mathrm{sup}}%
\mathrm{~}\frac{1}{2}e^{-c}\left( v_{\ast }-\frac{n_{\ast }-1}{c}%
C_\delta\right)$.
\end{proof}

\noindent By (\ref{irr}) and (\ref{NN0}) we know that $P$ is irreducible.
Hence, we can apply here Holley's theorem, see \cite[Theorem 3.7]{H}, and
obtain the following statement.

\begin{proposition}
\label{Holpn} Assume that the inequality%
\begin{equation}
P\left( \vartheta _{k}=1|\vartheta _{k^{\prime }}=\beta _{k^{\prime }},\
k^{\prime }\in \mathsf{L}\setminus \{k\}\right) \geq Q(\vartheta _{k}=1)
\label{Hol}
\end{equation}
holds for each $k\in \mathsf{L}$ and $\beta \in \{0,1\}^{\mathsf{L}\setminus
\{k\}}$. Then $Q\prec P$.
\end{proposition}

Recall that $P$ is determined by $P_{\Lambda }^{\hat{\xi}}$ with a fixed $%
\hat{\xi}\in \widehat{\Gamma }(n_{\ast },a)$. For this $\hat{\xi}$, and $k$
and $\beta $ as in (\ref{Hol}), we pick $\hat{\eta}\in \Gamma ^{\mathrm{t}%
}(X,S)$ such that: (a) $\hat{\eta}_{\Lambda ^{c}}=\hat{\xi}_{\Lambda ^{c}}$;
(b) $\vartheta _{k^{\prime }}(\eta )=\beta _{k^{\prime }}$ for each $%
k^{\prime }\in \mathsf{L}\setminus \{k\}$. Then
\begin{equation}
P\left( \vartheta _{k}=1|\vartheta _{k^{\prime }}=\beta _{k^{\prime }},\
k^{\prime }\in \mathsf{L}\setminus \{k\}\right) =P_{\Lambda }^{\hat{\xi}%
}(N(\gamma _{k})\geq n_{\ast }|\hat{\eta}),  \label{Hol1}
\end{equation}%
Observe that the conditional measure $P_{\Lambda }^{\hat{\xi}}(\cdot |\hat{%
\eta})$ can be obtained in the form
\begin{equation}
P_{\Lambda }^{\hat{\xi}}(d\gamma |\hat{\eta})=\int_{S^{\eta _{\Lambda
\setminus \Xi _{k}}}}P_{\Xi _{k}}^{\hat{\eta}}(d\gamma )\chi _{\eta
_{\Lambda \setminus \Xi _{k}}}(d\sigma _{\eta _{\Lambda \setminus \Xi
_{k}}}),  \label{int1}
\end{equation}%
see (\ref{di2}).

\begin{lemma}
\label{lemma1} Let $k$ and $\hat{\eta}$ be as in (\ref{Hol}), (\ref{Hol1}), (%
\ref{int1}). Then for any $n_{\ast }\in \mathbb{N}$ and $q_{0}\in (0,1)$
there exists $z_{c}>0$ such that
\begin{equation}
P_{\Xi _{k}}^{\hat{\eta}}(N(\gamma _{k})\geq n_{\ast })\geq q_{0}
\label{zc1}
\end{equation}%
for all $z>z_{c}$.
\end{lemma}

\begin{proof}
Let $I_{n}$ be the indicator function of the set $\{\gamma :N(\gamma
_{k})=n\}$, $n\in \mathbb{N}$. Set also $\omega _{n}=P_{\Xi _{k}}^{\hat{\eta}%
}(N(\gamma _{k})=n)$. By (\ref{LPM}, (\ref{cond-en-H}), (\ref{di3A}), and (%
\ref{g-fun}) for $n<n_{\ast }$ we get
\begin{eqnarray*}
\omega _{n+1} &=&\frac{1}{n+1}\int_{\Gamma (X,S)}\left( \sum_{x\in \gamma
_{k}}I_{n+1}(\gamma )\right) \Pi _{\Xi _{k}}(d\hat{\gamma}|\hat{\eta}) \\%
[0.2cm]
&=&\frac{z}{n+1}\int_{\Gamma (X,S)}I_{n}(\gamma )g(\hat{\gamma})\Pi _{\Xi
_{k}}(d\hat{\gamma}|\hat{\eta}) \\[0.2cm]
&\geq &\frac{zg_{\ast }}{n+1}\omega _{n}\geq zt_{\ast }\omega _{n},\qquad
t_{\ast }:=g_{\ast }/n_{\ast },
\end{eqnarray*}%
where we have taken into account that $n+1\leq n_{\ast }$ and used (\ref%
{est3}). The latter estimate readily yields
\begin{equation*}
\sum_{n=0}^{n_{\ast }-1}\omega _{n}\leq \frac{\omega _{n_{\ast }}}{zt_{\ast
}-1}\leq \frac{1}{zt_{\ast }-1}\sum_{n\geq n_{\ast }}\omega _{n}.
\end{equation*}%
Taking into account that $\sum_{n\geq 0}\omega _{n}=1$ we obtain that%
\begin{equation*}
P_{\Xi _{k}}^{\hat{\eta}}(N(\gamma _{k})\geq n_{\ast })\geq 1-\frac{1}{%
zt_{\ast }}.
\end{equation*}%
Now we can set
\begin{equation}
z_{c}=\left( t_{\ast }(1-q_{0})\right) ^{-1},  \label{zc}
\end{equation}%
and (\ref{zc1}) follows.
\end{proof}

\noindent \textit{Proof of Lemma \ref{5lm}.} For $z>z_{c}$ given in (\ref{zc}%
), we have $P_{\Xi _{k}}^{\hat{\eta}}(N(\gamma _{k}\geq n_{\ast }))\geq q_0$%
, which by (\ref{int1}) and (\ref{Hol1}) yields (\ref{Hol}) and hence $%
Q\prec P $ by Proposition \ref{Holpn}.
%TCIMACRO{\TeXButton{\hfill}{\hfill}}%
%BeginExpansion
\hfill%
%EndExpansion
$\square $

%\section{Appendix: Choosing the parameters}

%For readers' convenience, we collect here the relations between the values
%of the main parameters that have been used throughout the paper.

%\begin{itemize}
%\item Choose $a>0$ to satisfy (\ref{di5a}) and $q$ such that (\ref{di8a})
%holds.

%\item Choose $p\in (q^{\mathrm{s}}(d),1)$, then $n_{\ast }$ such that both
%estimates in (\ref{di11}).

%\item Choose $z_c$ as in (\ref{zc}). Alternatively, take $z_c = 1/t_*(1- q^{%
%\mathrm{bs}}(d))$. Then, for a given $z>z_c$, pick $p\in (q^{\mathrm{bs}%
%}(d),1)$ such that $z> 1/t_*(1-p)$.
%\end{itemize}

\vskip.2cm \noindent \textbf{Acknowledgment:} This work was financially
supported by the DFG through the SFB 701: `Spektrale Strukturen und
Topologische Methoden in der Mathematik' and by the European Commission
under the project STREVCOMS PIRSES-2013-612669.

%\section{Acknowledgment}

\end{document}